\documentclass[14pt, twoside]{article}
 \pdfoutput=1
\usepackage{float}
\usepackage{mathrsfs}
\usepackage{amsfonts}
\usepackage{subfigure}
\usepackage{diagbox}
\usepackage{amsfonts}
\usepackage{amsmath}
\usepackage{amssymb}
\usepackage{multicol}
\usepackage{graphics}
\usepackage{stmaryrd}
\usepackage{cite}

\usepackage{amsmath,amssymb}
\usepackage{multicol}
\usepackage{multicol,graphics}
\usepackage{graphicx}
\usepackage{epstopdf}
\usepackage{epsfig}
\usepackage{pdfpages}
\usepackage{authblk}

\newtheorem{theorem}{Theorem}[section]

\newtheorem{remark}[theorem]{Remark}

\numberwithin{equation}{section}
\newenvironment{proof}[1][Proof]{\noindent\textbf{#1.} }{\hfill $\Box$}
\allowdisplaybreaks

\makeatletter
\makeatletter
    
    \newcommand{\Rmnum}[1]{\expandafter\@slowromancap\romannumeral #1@}
    \makeatother
\date{}

\begin{document}
\title{\textbf{Transmission Dynamics of COVID-19 Pandemic Non-pharmaceutical Interventions and Vaccination }~\thanks{%
Supported by NSF of China 11501269 and 11731005.}}
\author[1]{Bin-Guo WANG \thanks{{\tt
E-mail address:wangbinguo@lzu.edu.cn}}$^,$}
\author[2]{Shunxiang HUANG\thanks{{\tt
E-mail address:huangshunxiang@mail.iap.ac.cn}}$^,$\thanks{ Contributed equally to this article}}
\author[1]{Yongping XIONG }
\author[1]{Ming-Zhen XIN }
\author[2]{Jing LI}
\author[1]{Jiangqian ZHANG}
\author[1]{Zhihui Ma}

\affil[1]{\small\it School of Mathematics and Statistics, Lanzhou
University, Lanzhou, Gansu 730000, China}
\affil[2]{\small\it Institute of NBC Defense of PLA, Beijing 102205, China}
\maketitle

\begin{abstract}
Non-pharmaceutical interventions(NPIs) play an important role in the early stage control of  COVID-19 pandemic. Vaccination is considered to be the inevitable course to stop the spread of SARS-CoV-2. Based on the mechanism, a SVEIR COVID-19 model with vaccination and NPIs is proposed. By means of the basic reproduction number $R_{0}$, it is shown that the disease-free equilibrium is globally attractive if $\mathscr{R}_{0}<1$, and COVID-19 is uniform persistence if $\mathscr{R}_{0}>1$. Taking Indian dates for example in the numerical simulation, we find that our dynamical
results fits well with the statistical dates. Consequently, we forecast the spreading trend of COVID-19 pandemic in India. Furthermore, our results imply that improving the intensity of NPIs
 will greatly reduce the number of confirmed cases. Especially, NPIs are indispensable even if all the people were vaccinated when the efficiency of vaccine is relatively low. By simulating the relation ships of the basic reproduction number $\mathscr{R}_{0}$, the vaccination rate and the efficiency of vaccine, we find that it is impossible to achieve the herd immunity without NPIs when the efficiency of vaccine is lower than $76.9\%$. Therefore, the herd immunity area is defined by the evolution of relationships between the vaccination rate and the efficiency of vaccine. In the study of two patchy, we give the conditions for India and China to be open to navigation. Furthermore, an appropriate dispersal of population between India and China is obtained. A discussion completes the paper.


\end{abstract}
\section{Introduction}\indent
\noindent

Coronavirus disease 2019 (COVID-19), caused by a novel virus of the coronavirus genus (SARS-CoV-2), has been spreading globally. As of 28 June, 2021, there have been 180,817,269 confirmed cases of COVID-19, including 3,923,238 deaths\cite{WHOrenshu}. The ongoing COVID-19 pandemic
has caused a Once-in-a-Century global crisis\cite{WHOrenshu11}.
Despite scientists worldwide racing to develop antiviral
drugs, curative treatments are unavailable at the time of writing.
the global economy is experiencing the worst plunge in recent
history amid fears of further deterioration of the COVID-19
situation\cite{XiaopanGao}.

Non-pharmaceutical interventions (NPIs) such as quarantine,
isolation, and social distancing play an important role in the control of SARS-CoV-2. Since the outbreak of COVID-19 was
first detected in December 2019 in Wuhan, China\cite{Liq}, many authors have discussed the effects of various measures on the control of COVID-19 pandemic. For example, Tian et al. \cite{HuaiyuTian} suggested that the
Wuhan travel ban or the national emergency
response would have decreased to 744,000
($\pm$ 156,000) confirmed COVID-19 cases outside
Wuhan. Since the airborne transmission by droplets and aerosols is important for the spread of viruses, face masks are a well-established preventive measure. Cheng et al. \cite{chengy} found that most environments and contacts are under conditions of low virus abundance (virus-limited) where surgical masks are effective at preventing virus spread. More advanced masks and other protective equipments are required in potentially virus-rich indoor environments including medical centers and hospitals. Hu et al. \cite{ShixiongHu} based on individual
records of 1178 potential SARS-CoV-2 infectors and their 15,648 contacts in Hunan, China. Their results showed that SARS-CoV-2 susceptibility to infection increases with age, while transmissibility is
not significantly different between age groups and between symptomatic and asymptomatic
individuals. Contacts in households and exposure to first-generation cases are associated
with higher odds of transmission. Considered the infectivity of
individuals with, and susceptibility to, SARS-CoV-2 infection
differs by age, schools were closed in
the early months of the pandemic in most countries\cite{Lancker, Educational}, so that the
low proportion of cases notified in young individuals\cite{Sinha} could be
attributed to a low probability of developing symptoms\cite{Poletti, pollan}, a low
susceptibility to infection\cite{zhangj,jingq,wuj}, and/or few contact opportunities
relative to other age groups. Senapati et al. \cite{Chattopadhyay} investigate that higher intervention
effort is required to control the disease outbreak within a shorter period of time in India.

The transmission mechanism of SARS-CoV-2
makes it more difficult to fight against the disease. As far as the situation is concerned, vaccines are considered to be the most effective defense to control the disease completely. Various deployment strategies were being proposed to increase population immunity levels to SARS-CoV-2. In Saad-Roy et al. \cite{Saad-Roy}, the authors explored three scenarios of selection and found that a one-dose policy may increase the potential for antigenic evolution under certain conditions of partial population immunity. At same time, they highlighted the critical need to test viral loads and quantify immune responses after one vaccine dose, and to ramp up vaccination efforts throughout the world. In consideration of limited initial supply of SARS-CoV-2 vaccine, Bubar et al. \cite{Bubar} used a mathematical model to compare five
age-stratified prioritization strategies.  Following some of the WHO-SAGE recommendations, Acu$\tilde{n}$a-Zegarraa et al. formulated an optimal control problem with mixed constraints to describe vaccination schedules\cite{Zegarraa}.

According to the transmission mechanism of COVID-19, combining NPIs and vaccines, the population is divided into the following categories: susceptible individuals (S), vaccinated individuals(V), exposed individuals (E), infectious individuals (I) and recovered/removed individuals (R). Hence, a SVEIR model for COVID-19 is proposed. The SEIR model and its updates were successful in predicting the SARS epidemic in 2003-2004\cite{Gumel}, the H1N1 influenza pandemic in 2009\cite{Hwang}, and the MERS epidemic in 2012-2015\cite{Chowell}. Recent studies have tried to use the SEIR model to predict the spread of COVID-19 in China. Raed et al. (2020) found that there would be approximately 21,022 (95\% CrI 11,090-33,490) total infections by February 22 in Wuhan \cite{Cummings}, while Wu et al. estimated that there would be 75,815 COVID-19 cases (95\% CrI 37,304-130,330) in Wuhan by January 25, 2020\cite{Leung}.

The basic reproduction number (ratio) $\mathscr{R}_{0}$ is a crucial threshold parameter in the study of disease transmission. In epidemiology, it is defined as the expected number of secondary cases produced by a single (typical) infection in a completely susceptible population and is used to measure the infection potential of an infectious disease
\cite{pvan, diekmann0}. At the beginning of the transmission of coronavirus,
based on likelihood and model analysis, Tang et al. \cite{Tangbiao} revealed that the basic reproduction number may be as high as 6.47, which showed that COVID-19 is highly infectious. By means of the basic reproduction number, Bubar et al. \cite{Bubar} found a highly mitigated spread during vaccine rollout. Riley et al. \cite{Riley} used a model of constant exponential
growth and decay, and quantified this
fall and rise in prevalence in terms of halving
and doubling times and the basic reproduction number. Noting that an important quantity in epidemiological models, the basic reproduction number, Cuevas-Maraver et al.\cite{Maraver} discussed in the realm of the model what consequences different additional intervention measures would have had at the level of deaths and of cumulative infections.

In the face of vaccine dose shortages and logistical challenges, how to deploy the strategies of NPIs and viccination to increase population immunity levels to SARS-CoV-2 have not been fully considered in most of the above published. Mathematical models and field observations show that population
dispersal can exert strong pressure on many infectious diseases\cite{gao, liupeter, wangwendizhao}. Therefore, we choose a spatially discrete environment consisting of $n$ patches, where a patch may represent a country or a city, and population movements between patches. Based on the infection mechanism of SARS-CoV-2, a novel susceptible-asymptomatic-symptomatic-recovered model with NPIs and viccination (SASR) model in a patchy environment is proposed in Section 2. On the basis of the existence of the disease-free equilibrium, the definition and computation formulae of the basic reproduction number $\mathscr{R}_{0}$ are established. Drawing support from the basic reproduction number, the extinction and uniform persistence are shown in Section 3. The numerical simulation results not only consider the intensity of the intervention, the vaccination rate and the efficiency of vaccine but also incorporate the relationship between them and the basic reproducing number so as to get the control strategy of COVID-19 in Section 4. A brief discussion completes the paper.

\section{ SASR Compartmental Model}\label{apsec2}
\noindent

A matrix $M$ is said to be nonnegative if all entries of $M$ are nonnegative. If all off-diagonal entries of $M$ are nonnegative, then we say $M$ is cooperative.

A $n\times n$ matrix $A=(a_{ij})_{n\times n}$ is said to be  irreducible  if for every nonempty, proper subset $I$ of the set $N=\{1,2,\cdots,n\}$, there is an $i\in I$ and $j\in J=N\backslash I $ such that $|a_{ij}|>0$.

Let $S_{i}$ be  the number of
susceptible individuals in patch $i$, $V_{i}$ be  the number of the vaccinated individuals, $E_{1i}$ be  the number of
the exposed individuals who are not contagious in the early stages in patch $i$, $E_{2i}$ be  the number of
the exposed individuals who can infect the susceptible in patch $i$, and $I_{1i}$ be  the number of
the infectious individuals in patch $i$ who are contagious, $I_{2i}$ be  the number of
the infectious individuals in patch $i$ who are not contagious, $R_{i}$ be  the number of
recovery individuals in patch $i$, $N_i$ be the total population in patch $i$, that is, $N_i=S_i+V_i+E_{1i}+E_{2i}+I_{1i}+I_{2i}+R_i$. Furthermore, we suppose that $N_i$ is constant. The population growth process can be described as in Figure \ref{F0}.

\begin{figure}[H]
\centering
\subfigure{
\includegraphics[height=2.0in,width=5.5in]{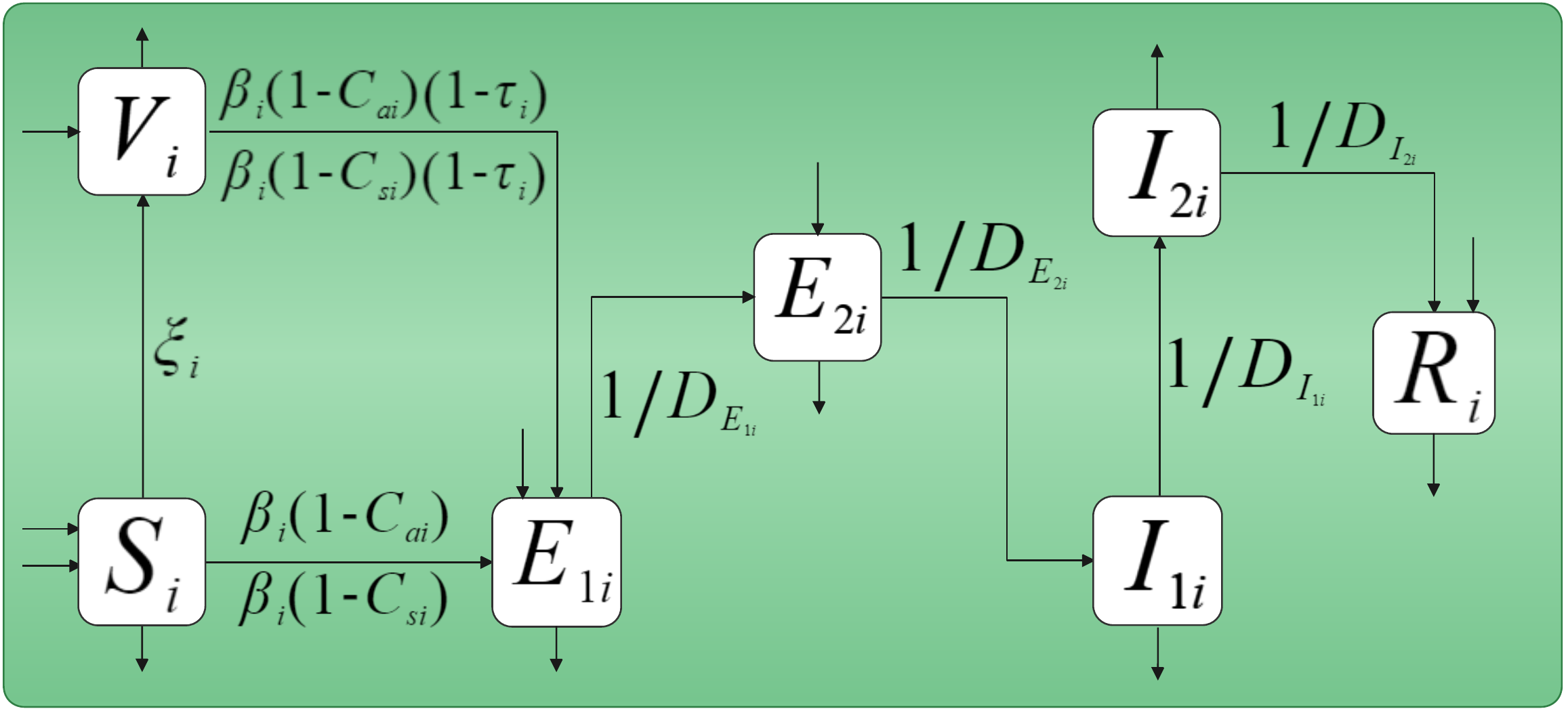}}
\caption{Compartmental diagram of COVID-19 transmission dynamics}
\label{F0}
\end{figure}
The model is given by an autonomous system of ordinary differential equations
\begin{equation}\label{2.14}
\begin{array}{ll}
\frac{dS_{i}}{dt}=\Lambda_{i}-\beta_{i}(1-C_{ai})\frac{S_{i}E_{2i}}{N_{i}}-\beta_{i}(1-C_{si})\frac{S_{i}I_{1i}}{N_i}-\xi_i S_{i}-\mu_{i}S_{i}+
\sum\limits_{j=1}^{n}A_{ij}S_{j}, \\
\frac{dV_{i}}{dt}=\xi_i S_{i}-\beta_{i}(1-C_{ai})(1-\tau_i)\frac{V_{i}E_{2i}}{N_i}-\beta_{i}(1-C_{si})(1-\tau_i)\frac{V_{i}I_{1i}}{N_i}-\mu_{i}V_{i}
+\sum\limits_{j=1}^{n}A_{ij}V_{j} ,\\
\frac{dE_{1i}}{dt}=\beta_{i}(1-C_{ai})\frac{S_{i}E_{2i}}{N_i}+\beta_{i}(1-C_{si})\frac{S_{i}I_{1i}}{N_i}+\beta_{i}
(1-C_{ai})(1-\tau_i)\frac{V_{i}E_{2i}}{N_i}+\beta_{i}(1-C_{si})(1-\tau_i)\frac{V_{i}I_{1i}}{N_i}\\
\ \ \ \ \ \ -\frac{E_{1i}}{D_{E_{1i}}}-\mu_{i}E_{1i}+\sum\limits_{j=1}^{n}B_{ij}E_{1j} ,\\
\frac{dE_{2i}}{dt}=\frac{E_{1i}}{D_{E_{1i}}}-\frac{E_{2i}}{D_{E_{2i}}}-\mu_{i}E_{2i}-d_{i}E_{2i}+
\sum\limits_{j=1}^{n}C_{ij}E_{2j},\\
\frac{dI_{1i}}{dt}=\frac{E_{2i}}{D_{E_{2i}}}-\frac{I_{1i}}{D_{I_{1i}}}-\mu_iI_{1i}-d_{i}I_{1i}, \\
\frac{dI_{2i}}{dt}=\frac{I_{1i}}{D_{I_{1i}}}-\frac{I_{2i}}{D_{I_{2i}}}-\mu_iI_{2i},\\
\frac{dR_{i}}{dt}=\frac{I_{2i}}{D_{I_{2i}}}-\mu_{i}R_{i}+\sum\limits_{j=1}^{n}D_{ij}R_{j}, \\
\end{array}
\end{equation}
where $\Lambda_i$ is the recruitment rate of susceptible class in patch $i$,  $C_{ai}$ and $C_{si}$ denote the intensity of NPIs for incubation with infectiousness and infection with infectiousness individuals in patch $i$, respectively. $\beta_i$ denotes the effective contact rate in patch $i$, $\xi_i$ denotes the vaccination coverage rate in patch $i$,
 $\mu_i$ is the natural
death rate of the population in patch $i$, $D_{E_{1i}}$ and $D_{E_{2i}}$ are lengths of the incubation with non-infectiousness and incubation with infectiousness in patch $i$, respectively. $0\leq\tau_i\leq1$ denotes the vaccine efficacy in patch $i$($\tau_i=1$ represents a vaccine that
offers 100\% protection against infection, $\tau_i= 0$ models a vaccine that offers no
protection at all).
 $D_{I_{1i}}$ and $D_{I_{2i}}$ are lengths of the infection with infectiousness and infection with non-infectiousness in patch $i$, respectively. $d_{i}$ denotes the disease-induced mortality
rate. $A_{ij}$ stand for the dispersal rate form patch $j$ to patch $i$ of the susceptible and the vaccination. $B_{ij}$, $C_{ij}$ and $D_{ij}$ stand for the dispersal rates from patch $j$ to patch $i$ of the incubation with non-infectiousness, incubation with infectiousness, and recovery individuals, respectively. Biologically, we could suppose that
the number of total human population in patch $i$  stabilizes at $N_i>0$.

Moreover, we assume
    \begin{enumerate}
	
	\item[(A1)] $A_{ij}, B_{ij},C_{ij}$ and $D_{ij}$ are nonegative constant,$~\forall 1\leq i\neq j\leq n$, and $[{A_{ij}}]_{n\times n}$, $[{B_{ij}}]_{n\times n}$, $[{C_{ij}}]_{n\times n}$ and $[{D_{ij}}]_{n\times n}$ are irreducible.
	
	\item[(A2)] $\sum\limits_{j=1}^{n}A_{ji}=\sum\limits_{j=1}^{n}B_{ji}=\sum\limits_{j=1}^{n}C_{ji}=\sum\limits_{j=1}^{n}D_{ji}=0
,~\forall i=1,\cdots,n.$

	 \end{enumerate}

Note that (A1) assures that the immigration always occurs between two groups
which are the arbitrary separation of n patches; (A2) means that deaths and births are neglected during the dispersal process.

For simplicity, set $\psi_t(x^0)$ be the solution of (\ref{2.14}) with initial date $\psi_0(x^0)=x^0\in \mathbb{R}_+^{7n}$.  By \cite[Theorem 2.1]{gaoruan}, we have the following.

\begin{theorem}
For any $x^0\in \mathbb{R}_+^{7n}$,  system $(\ref{2.14})$ has a unique nonnegative solution $\psi_t(x^0)$ with initial value $\psi_0(x^0)=x^0$, and all solutions are ultimately bounded and uniformly bounded.
	\end{theorem}

\section{Basic Reproduction Number}
\noindent

In this section, we establish the definition and computation formulae of the basic reproduction number for system $(\ref{2.14})$.

We first consider the disease-free solution of system $(\ref{2.14})$. Let $E_{1i}=E_{2i}=I_{1i}=I_{2i}=0$, $i=1,\cdots,n,$ then we have
	\begin{equation}\label{diseasefree}
\begin{array}{ll}
	\frac{dS_i}{dt}=\Lambda_i-\xi_i S_{i}-\mu_iS_i+\sum\limits_{j=1}^{n}A_{ij}S_j,\\
\frac{dV_{i}}{dt}=\xi_i S_{i}-\mu_{i}V_{i}+\sum\limits_{j=1}^{n}A_{ij}V_{j}.\\
\end{array}
\end{equation}

By the similar arguments to those in \cite{wang_basic_2013}, system (\ref{diseasefree}) has a positive equilibrium $(S^*,V^*)=(S_1^*,S_2^*,\cdots, S_n^*,V_1^*,V_2^*,\cdots, V_n^*)$, which is globally attractive. Linearizing system ($\ref{2.14}$) at the disease-free equilibrium $(S^*,V^*,0,0,0,0,0)$, we get
	\begin{equation}\label{linearequanttion}
	\begin{array}{ll}
	 \frac{dE_{1i}}{dt}=\beta_{i}(1-C_{ai})\frac{S^*_{i}E_{2i}}{N_i}+\beta_{i}(1-C_{si})\frac{S^*_{i}I_{1i}}{N_i}+
\beta_{i}(1-C_{ai})(1-\tau_i)\frac{V^*_{i}E_{2i}}{N_i}+\beta_{i}(1-C_{si})(1-\tau_i)\frac{V^*_{i}I_{1i}}{N_i}\\
~~~~~~~~~~-\frac{E_{1i}}{D_{E_{1i}}}-\mu_{i}E_{1i}+\sum\limits_{j=1}^{n}B_{ij}E_{1j} ,\\
\frac{dE_{2i}}{dt}=\frac{E_{1i}}{D_{E_{1i}}}-\frac{E_{2i}}{D_{E_{2i}}}-\mu_{i}E_{2i}-d_{i}E_{2i}+\sum\limits_{j=1}^{n}C_{ij}E_{2j},\\
\frac{dI_{1i}}{dt}=\frac{E_{2i}}{D_{E_{2i}}}-\frac{I_{1i}}{D_{I_{1i}}}-\mu_iI_{1i}-d_{i}I_{1i}. \\
	\end{array}
	\end{equation}
	Define
	\begin{equation}\nonumber
	 F_1=
	\left(
	\begin{array}{cccc}
	\beta_{1}(1-C_{a1})\frac{S^*_1+(1-\tau_1)V^*_{1}}{N_1} & 0 & \cdots & 0\\
     0 &\beta_{2}(1-C_{a2})\frac{S^*_2+(1-\tau_2)V^*_{2}}{N_2} & \cdots & 0\\
	\vdots & \vdots & \ddots & \vdots \\
	0 & 0 & \cdots & \beta_{n}(1-C_{an})\frac{S^*_n+(1-\tau_n)V^*_{n}}{N_n}
	\end{array}
	\right),
	\end{equation}
	\begin{equation}\nonumber
	F_2=
	\left(
	\begin{array}{cccc}
	\beta_{1}(1-C_{s1})\frac{S^*_1+(1-\tau_1)V^*_{1}}{N_1} & 0 & \cdots & 0\\
     0 &\beta_{2}(1-C_{s2})\frac{S^*_2+(1-\tau_2)V^*_{2}}{N_2}& \cdots & 0\\
	\vdots & \vdots & \ddots & \vdots \\
	0 & 0 & \cdots & \beta_{n}(1-C_{sn})\frac{S^*_n+(1-\tau_n)V^*_{n}}{N_n}
	\end{array}
	\right),
	\end{equation}
	\begin{equation}\nonumber
	 V_1=
	\left(
	\begin{array}{cccc}
	\frac{1}{D_{E_{11}}}+\mu_1-B_{11}& -B_{12} & \cdots & -B_{1n}(t)\\
	-B_{21} & \frac{1}{D_{E_{12}}}+\mu_2-B_{22} & \cdots & -B_{2n}\\
	\vdots & \vdots & \ddots & \vdots \\
	-B_{n1}& -B_{n2} & \cdots & \frac{1}{D_{E_{1n}}}+\mu_n-B_{nn}
	\end{array}
	\right),
	\end{equation}
	
	\begin{equation}\nonumber
	 V_2=
	\left(
	\begin{array}{cccc}
	\frac{1}{D_{E_{21}}}+\mu_1+d_{1}-C_{11} & -C_{12} & \cdots & -C_{1n}\\
	-C_{21} & \frac{1}{D_{E_{22}}}+\mu_2+d_{2}-C_{22} & \cdots & -C_{2n}\\
	\vdots & \vdots & \ddots & \vdots \\
	-C_{n1} & -C_{n2} & \cdots & \frac{1}{D_{E_{2n}}}+\mu_n+d_{n}-C_{nn}
	\end{array}
	\right),
	\end{equation}
	
	\[ V_3=
	\begin{array}{c}
	diag(\frac{1}{D_{I_{11}}}+\mu_1, \frac{1}{D_{I_{12}}}+\mu_2 \cdots, \frac{1}{D_{I_{1n}}}+\mu_n),
	\end{array}
	\]
	\[ V_4=
	\begin{array}{c}
	diag(\frac{1}{D_{E_{11}}}, \frac{1}{D_{E_{12}}}, \cdots, \frac{1}{D_{E_{1n}}})
	\end{array} \]
and
	\[ V_5=
	\begin{array}{c}
	diag(\frac{1}{D_{E_{21}}}, \frac{1}{D_{E_{22}}}, \cdots, \frac{1}{D_{E_{2n}}}).
	\end{array}
	\]
	Let
	\begin{equation}\nonumber
	 Y = \left(
	\begin{array}{ccc}
	0 &  F_1 &  F_2\\ 0 & 0 & 0 \\ 0 & 0 & 0
	\end{array}\right),  Z = \left(
	\begin{array}{ccc}
	 V_1 & 0 & 0\\ - V_4 &  V_2 & 0\\ 0 & - V_5 &  V_3
	\end{array}\right),
	\end{equation}
then (\ref{linearequanttion}) can be written by
	\begin{equation}\label{3.6}
	\frac{du}{dt}=(Y-Z)u.
	\end{equation}

Motivated by the concept of next generation matrices introduced in
\cite{pvan, diekmann0}, we define the basic reproduction number of system $(\ref{2.14})$ as
\begin{equation}\label{jbzsh}
 \mathscr{R}_{0}:=\rho( Y  Z^{-1}),
\end{equation}
where $\rho(A)$ denotes the spectral radius of a matrix $A$.

Let
$s(A)$ denotes the maximum real part of all the eigenvalues of the matrix $A$ (the spectral
abscissa of $A$). By a similar arguments to those in \cite{pvan}, we have the following statements.
\begin{theorem}
$\mathscr{R}_0-1$ has the same sign as $s(Y-Z)$.
\end{theorem}

The following two theorems give a threshold-type result on the extinction and uniform persistence of COVID-19 in terms of $\mathscr{R}_{0}$.

\begin{theorem}\label{Theorem4.1}
		Assume $(A1)$-$(A2)$ holds and $\mathscr{R}_{0}<1$, then the disease-free equilibrium $\mathcal{E}^*=(S^*,V^*,0,0,0,0,0)$ of system $(\ref{2.14})$ is globally attractive.
\end{theorem}
\begin{proof}
It is easy to see that $S(t)$ and $V(t)$ satisfy
\begin{equation*}
\begin{split}
	\frac{dS_i}{dt}\leq\Lambda_i-\xi_i S_{i}-\mu_iS_i+\sum\limits_{j=1}^{n}A_{ij}S_j
\end{split}
\end{equation*}
and
\begin{equation*}
\begin{split}
\frac{dV_{i}}{dt}\leq\xi_i S_{i}-\mu_{i}V_{i}+\sum\limits_{j=1}^{n}A_{ij}V_{j},
\end{split}
\end{equation*}
respectively.
Since system (\ref{diseasefree}) has a unique positive constant solution ($S_1^*,\cdots,S_n^*, V_1^*,\cdots,V_n^*$), which is global asymptotically stable,  by the comparsion theorem, for any $x^0\in \mathbb{R}_+^{7n}$ and $ \varepsilon >0$, there exists $t_0>0$ such that \[S_i(t)\leq S_{i}^*+\varepsilon, V_i(t)\leq V_i^*+\varepsilon,   ~~\forall t\geq t_0,\quad i=1,\cdots,n.\]
It then follows that
		\begin{equation*}\label{4.4}
		\begin{array}{ll}
		\frac{dE_{1i}}{dt}&\leq\beta_{i}(1-C_{ai})\frac{(S_{i}^*+\epsilon) E_{2i}}{N_i}+\beta_{i}(1-C_{si})\frac{(S^*_{i}+\epsilon)I_{1i}}{N_i}+\beta_{i}
(1-C_{ai})(1-\tau_i)\frac{(V^*_{i}+\epsilon)E_{2i}}{N_i}\\&~~~+\beta_{i}(1-C_{si})(1-\tau_i)\frac{(V_{i}^*+\epsilon)I_{1i}}{N_i} -\frac{E_{1i}}{D_{E_{1i}}}-\mu_{i}E_{1i}+\sum\limits_{j=1}^{n}B_{ij}E_{1j} ,\\
		 \frac{dE_{2i}}{dt}&\leq\frac{E_{1i}}{D_{E_{1i}}}-\frac{E_{2i}}{D_{E_{2i}}}-\mu_{i}E_{2i}+\sum\limits_{j=1}^{n}C_{ij}E_{2j},\\
		\frac{dI_{1i}}{dt}&\leq\frac{E_{2i}}{D_{E_{2i}}}-\frac{I_{1i}}{D_{I_{1i}}}-\mu_iI_{1i}.\\
		\end{array}
		\end{equation*}
we consider the following system
	\begin{equation}
		\begin{array}{ll}\label{4.5}
		\frac{dE_{1i}}{dt}&=\beta_{i}(1-C_{ai})\frac{(S_{i}^*+\epsilon) E_{2i}}{N_i}+\beta_{i}(1-C_{si})\frac{(S^*_{i}+\epsilon)I_{1i}}{N_i}+\beta_{i}
(1-C_{ai})(1-\tau_i)\frac{(V^*_{i}+\epsilon)E_{2i}}{N_i}\\&~~~+\beta_{i}(1-C_{si})(1-\tau_i)\frac{(V_{i}^*+\epsilon)I_{1i}}{N_i} -\frac{E_{1i}}{D_{E_{1i}}}-\mu_{i}E_{1i}+\sum\limits_{j=1}^{n}B_{ij}E_{1j} ,\\
		 \frac{dE_{2i}}{dt}&=\frac{E_{1i}}{D_{E_{1i}}}-\frac{E_{2i}}{D_{E_{2i}}}-\mu_{i}E_{2i}+\sum\limits_{j=1}^{n}C_{ij}E_{2j},\\
		\frac{dI_{1i}}{dt}&=\frac{E_{2i}}{D_{E_{2i}}}-\frac{I_{1i}}{D_{I_{1i}}}-\mu_iI_{1i}.\\
		\end{array}
		\end{equation}

Let $\lambda_{i}=\beta_{i}(1-C_{ai})(\frac{(N_i^*+2\epsilon)-\tau_i(V^*_{i}+\epsilon)}{N_i})$ and $\eta_{i}=\beta_{i}(1-C_{si})(\frac{(N_i^*+2\epsilon)-\tau_i(V^*_{i}+\epsilon)}{N_i})$, $i=1,2,...,n$.
Denote

	\begin{equation}\nonumber
	 F^{\epsilon}_1=
	\left(
	\begin{array}{cccc}
	\lambda_{1} & 0 & \cdots & 0\\
     0 &\lambda_{2} & \cdots & 0\\
	\vdots & \vdots & \ddots & \vdots \\
	0 & 0 & \cdots & \lambda_{n}
	\end{array}
	\right),
	F^{\epsilon}_2=
	\left(
	\begin{array}{cccc}
	\eta_{1} & 0 & \cdots & 0\\
     0 &\eta_{2}& \cdots & 0\\
	\vdots & \vdots & \ddots & \vdots \\
	0 & 0 & \cdots \eta_{n}
	\end{array}
	\right).
\end{equation}
Let
	\begin{equation}\nonumber
	 Y^{\epsilon} = \left(
	\begin{array}{ccc}
	0 &  F^{\epsilon}_1 &  F^{\epsilon}_2\\ 0 & 0 & 0 \\ 0 & 0 & 0
	\end{array}\right).
	\end{equation}
Since $\mathscr{R}_{0}<1$, then $ s(-Z+Y)<0$.  By the continuity of spectral bound,
there exists a sufficiently small $\epsilon_1>0$ such that $s(-Z+Y^{\epsilon})<0$ for $0<\epsilon<\epsilon_1$, which implies that the trivial solution of the system (\ref{4.5}) is globally asymptotically stable. By the comparison theorem of ordinary differential equation, we deduce that $E_{1i}\to 0, E_{2i}\to 0,I_{1i}\to 0$ as $t\to\infty$, $ \forall i=1,2,\cdots, n.$  It then follows that system $(\ref{diseasefree})$ is the limiting system of $S_i,V_i$ equation in system (\ref{2.14}).
 We also could get that $I_2,R$ equation admit the limiting system
\begin{equation}\label{limiting}
\begin{split}
\frac{dI_{2i}}{dt}&=-\frac{I_{2i}}{D_{I_{2i}}}-\mu_iI_{2i},\\
\frac{dR_{i}}{dt}&=\frac{I_{2i}}{D_{I_{2i}}}-\mu_{i}R_{i}+\sum\limits_{j=1}^{n}D_{ij}R_{j}. \\
\end{split}
\end{equation}
It is easy to see that the solutions in (\ref{limiting}) convergence to $(0,\cdots,0,0,\cdots,0)$. Finally, by the theory of asymptotically autonomous systems (see, e.g. \cite{CCT1995} ), we conclude that the solution of system $(\ref{2.14})$ converges to $(S_1^*,\cdots,S_n^*,V_1^*,\cdots,V_n^*,0,\cdots,0,0,\cdots,0)$. This confirms the global attractivity of $\mathcal{E}^*$ for system $(\ref{2.14})$ under the condition $\mathscr{R}_0<1$, and hence completes the proof.

\end{proof}

	\begin{theorem}\label{theorem4.2}
		If $\mathscr{R}_{0}>1$, then there exists $\tilde{\varepsilon} >0$ such that the solution $(S(t), V(t), E_{1}(t),  E_{2}(t),I_{1}(t),$\\ $I_{2}(t),R(t))$ of system $(\ref{2.14})$ with initial data $x^0$ in $\mathbb{R}_+^{7n}$ and $(E_{1}(0),E_{2}(0),I_{1}(0))>\hat0$ satisfies
		\[\liminf_{t \rightarrow \infty}E_{1i}(t)>\tilde{\varepsilon},\liminf_{t \rightarrow \infty}E_{2i}(t)>\tilde{\varepsilon},\liminf_{t \rightarrow \infty}I_{1i}(t)>\tilde{\varepsilon}, \quad\forall i=1,\cdots,n.\]
	\end{theorem}
\begin{proof}
Define
$$X=\mathbb{R}_+^{7n},$$
		\[X_0:=\{(S,V, E_1 ,E_2 ,I_1, I_2,R)\in X:E_{1i}>0,E_{2i}>0,I_{1i}>0, 1\leq i\leq n\}\]
		\[\partial X_0:=X \backslash  X_0\]
		Then $X_0$ and $\partial X_0$ are relatively open and closed in $\mathbb{R}^{7n},$ respectively. For any $ x^0\in X_0$, let $\psi_t(x^0)$ be the unique solution of system (\ref{2.14}) with initial data $x^0$.
		It is easy to see that $X_0$ is a positively invariant set.
		According to the arguments in Section 2, the solution of (\ref{2.14}) is ultimately bounded in $X$, which implies that $\psi_t:X\rightarrow X$ is point dissipative on $X$. It follows from \cite[Theorem 3.4.8]{Hale1988} that $\psi_t$ has a global compact attractor $\mathcal{A}$.
		
		Define \[M_{\partial}:=\{x^0\in \partial X_0:\psi_t(x^0)\in \partial X_0,\forall t\geq0\}\]
		and \[\mathcal{M}:=\{x^0\in X:x^0_{1i}=S_i^*,x^0_{2i}=V_i^*,x^0_{3i}= x^0_{4i}=x^0_{5i}=0,\forall i=1,\cdots,n\}.\]

We now show that \[M_{\partial}=\mathcal{M}.\]

For any $x^0\in \mathcal{M}$, the solution $\psi_t(x^0)$ satisfies $E_{1i}(t,x^0)=0,E_{2i}(t,x^0)=0,I_{1i}(t,x^0)=0,i=1,\cdots, n$ for all $t\geq 0$. Hence, $x^0\in M_{\partial}$ and $\mathcal{M}\subset M_{\partial}$.
		
		For any $x^0\in\partial X_0\backslash \mathcal{M}$, there is a $i^*$ such that $(x^0_{3i^*},x^0_{4i^*},x^0_{5i^*})=(E_{1i^*}(0),E_{2i^*}(0),I_{1i^*}(0))>(0,0,0)$.

\textbf{Case 1}
Let $E_{1i^*}(0)>0$. Since $ {B}_{ij}$ is cooperative,  the third equation of system (\ref{2.14}) satisfies
\[\frac{dE_{1i^*}}{dt}\geq-\frac{E_{1i^*}}{D_{E_{1i^*}}}-\mu_{i}E_{1i^*}+{B}_{i^*i^*}E_{1i^*}.\]
Furthermore, the fact that the matrix ${B}_{ij}$ is irreducible implies that there exists a $t_0>0$ such that $E_{1i}(t)>0$ for all $i=1,\cdots,n$ and $t\geq t_0$.

From the third  equation of system (\ref{2.14}), we can get $E_{2i}(t)>0$ $\forall i=1,\cdots,n$, $t\geq t_0+1$. Then, from the fourth  equation of system (\ref{2.14}), we deduct that $I_{1i}(t)>0$ for all $i=1,\cdots,n$ and $t\geq t_0+2$.

\textbf{Case 2}
Let $E_{2i^*}(0)>0$. By the fourth equation of system (\ref{2.14}), we have \[\frac{dE_{2i^*}}{dt}\geq-\frac{E_{2i^*}}{D_{E_{2i^*}}}-\mu_{i}E_{2i^*}+{C}_{i^*i^*}E_{2i^*},\]
which is deduced from the fact that $C_{ij}$ is cooperative. Thus, we can get $E_{2i^*}(t)>0,\forall t>0$. Now, the third equation satisfies
$$\frac{dE_{1i^*}}{dt}\geq \beta_{i^*}(1-C_{ai^*})\frac{S_{i^*}E_{2i^*}}{N_{i^*}}-\frac{E_{1i^*}}{D_{E_{1i^*}}}-\mu_{i}E_{1i^*}+{B}_{i^*i^*}E_{1i^*}.$$
It is easy to see that $E_{1i^*}(t)>0$ for $t>1$. By the arguments in \textbf{Case 1}, we can obtain that $(E_{1}(t),E_{2}(t),I_{1}(t))\gg (0,\cdots,0,0\cdots,0,0,\cdots,0)$ for all $t>t_0+3$.

\textbf{Case 3}
Let $I_{1i^*}(0)>0$. By the fifth equation of system (\ref{2.14}), we have \[\frac{dI_{1i^*}}{dt}\geq-\frac{I_{1i^*}}{D_{I_{1i^*}}}-\mu_iI_{1i^*}.\]

Hence, we can get $I_{1i^*}(t)>0,\forall t>0$. Now, the second equation of system (\ref{2.14}) satisfies
$$\frac{dE_{1i^*}}{dt}\geq\beta_{i^*}(1-C_{si^*})\frac{S_{i^*}I_{1i^*}}{N_{i^*}}-\frac{E_{1i^*}}{D_{E_{1i^*}}}-\mu_{i}E_{1i^*}+{B}_{i^*i^*}E_{1i^*}.$$
It is easy to see that $E_{1i^*}(t)>0$ for $t>1$. By the arguments in \textbf{Case 1}, we can obtain that $(E_{1}(t),E_{2}(t),I_{1}(t))\gg (0,\cdots,0,0\cdots,0,0,\cdots,0)$ for all $t>t_0+3$.

	Then $M_{\partial}\subset \mathcal{M}.$ Hence, $M_{\partial}=\mathcal{M}$.

We claim that $ W^s(\mathcal{M})\cap X_0= \emptyset$, where $W^s(\mathcal{M})$ is the stable manifold of $\mathcal{M}$.
Let $\bar{\lambda}_{i}=\beta_{i}(1-C_{ai})(\frac{(N_i^*-2\epsilon)-\tau_i(V^*_{i}-\epsilon)}{N_i})$ and $\bar{\eta}_{i}=\beta_{i}(1-C_{si})(\frac{(N_i^*-2\epsilon)-\tau_i(V^*_{i}-\epsilon)}{N_i})$, $\forall i=1,2,...,n$.
Denote

	\begin{equation}\nonumber
	 \bar{F}^{\epsilon}_1=
	\left(
	\begin{array}{cccc}
	\bar{\lambda}_{1} & 0 & \cdots & 0\\
     0 &\bar{\lambda}_{2} & \cdots & 0\\
	\vdots & \vdots & \ddots & \vdots \\
	0 & 0 & \cdots & \bar{\lambda}_{n}
	\end{array}
	\right),
	\bar{F}^{\epsilon}_2=
	\left(
	\begin{array}{cccc}
	\bar{\eta}_{1} & 0 & \cdots & 0\\
     0 &\bar{\eta}_{2}& \cdots & 0\\
	\vdots & \vdots & \ddots & \vdots \\
	0 & 0 & \bar{\eta}_{n}
	\end{array}
	\right).
\end{equation}
Let
	\begin{equation}\nonumber
	 \bar{Y}^{\epsilon} = \left(
	\begin{array}{ccc}
	0 &  F^{\epsilon}_1 &  F^{\epsilon}_2\\ 0 & 0 & 0 \\ 0 & 0 & 0
	\end{array}\right).
	\end{equation}

Since $\mathscr{R}_{0}>1$, then $ s(Y-Z)>0$.  By the continuity of spectral bound,
there exists a sufficiently small $\epsilon_1>0$ such that $s(\bar{Y}_{\epsilon}-Z)>0$ for $0<\epsilon\leq\epsilon_1$.

		$\textbf{Claim}.$ If $x^0\in X_0,$ then \[\limsup_{t \rightarrow \infty}d(\psi_t(x^0),\mathcal{M})\geq \epsilon_1\]

		On the contrary, we assume that there exists $\bar{x}^0\in X_0$ such that $\limsup\limits_{t \rightarrow \infty}d(\psi_t(\bar{x}^0),\mathcal{M})<\varepsilon_1$.
		It then follows that there exists $t_0>0$ such that
		 \[S_i^*-{\varepsilon}_1<S_i(t)<S_i^*+{\varepsilon}_1,V_i^*-{\varepsilon}_1<V_i(t)<V_i^*+{\varepsilon}_1\]
		for all $t\geq t_0 $ and $i=1,\cdots,n$.
			Hence, we have
\begin{equation}\label{4.8}
\begin{array}{ll}
\frac{dE_{1i}}{dt}&\geq\beta_{i}(1-C_{ai})\frac{(S_{i}^*-\varepsilon_1)E_{2i}}{N_i}+\beta_{i}(1-C_{si})\frac{(S^*_{i}-\varepsilon_1)I_{1i}}{N_i}+\beta_{i}
(1-C_{ai})(1-\tau_i)\frac{(V^*_{i}-\varepsilon_1)E_{2i}}{N_i}\\&~~~+\beta_{i}(1-C_{si})(1-\tau_i)\frac{(V^*_{i}-\varepsilon_1)I_{1i}}{N_i}-\frac{E_{1i}}{D_{E_{1i}}}-\mu_{i}E_{1i}+\sum\limits_{j=1}^{n}B_{ij}E_{1j} ,\\
\frac{dE_{2i}}{dt}&\geq\frac{E_{1i}}{D_{E_{1i}}}-\frac{E_{2i}}{D_{E_{2i}}}-\mu_{i}E_{2i}+\sum\limits_{j=1}^{n}{C}_{ij}E_{2j},\\
\frac{dI_{1i}}{dt}&\geq\frac{E_{2i}}{D_{E_{2i}}}-\frac{I_{1i}}{D_{I_{1i}}}-\mu_iI_{1i}.\\
\end{array}
\end{equation}
Since $-Z+\bar{Y}_{\epsilon}$ is irreducible and essentially nonnegative, it has a positive eigenvector associated with $s(-Z+\bar{Y}_{\epsilon})>0$.		
	By the comparison theorem of ordinary differential equations, we have  $\lim\limits_{t\to \infty }E_{1i}(t)=\infty,$$\lim\limits_{t\to \infty }E_{2i}(t)=\infty,$ $\lim\limits_{t\to \infty }I_{1i}(t)=\infty,$ a contradiction. The claim is proved.

		The set $M_{\partial}=\mathcal{M}$ is  an isolated invariant set and acyclic. By \cite[Theorem 4.6]{THR1993}, we conclude that system (\ref{2.14}) is uniformly persistent in $X_0$ whenever
$\mathscr{R}_{0}>1$. That is,  there is a $\tilde{\varepsilon}>0$ such  that \[\liminf_{t \rightarrow \infty}E_{1i}(t)>\tilde{\varepsilon},\liminf_{t \rightarrow \infty}E_{2i}(t)>\tilde{\varepsilon},\liminf_{t \rightarrow \infty}I_{1i}(t)>\tilde{\varepsilon}, \forall i=1,\cdots,n.\]
	This completes the proof.
\end{proof}

\begin{remark}\label{re}
\rm System $(\ref{2.14})$ considers the dynamics of COVID-19 model with NPIs and viccination. If we take $\tau_i=0$ in the above discussion, then system $(\ref{2.14})$ implies that NPIs for incubation with infectiousness and infection with infectiousness individuals is the only measure. Similarly, let $C_{si}=0$
and $C_{ai}=0$ in the above discussion, system $(\ref{2.14})$ implies that the vaccination is considered only.
\end{remark}

\section {Numerical simulation }
\noindent

By the actual dates showed in \cite{huangshunx}, let $\beta_i=0.7$, $D_{E_{1i}}=2.9$ Day, $D_{E_{2i}}=2.3$ Day, $D_{I_{1i}}=2.9$ Day, $D_{I_{2i}}=12$ Day, $\forall i=1,2,...,n$. In the case of one patchy, we take the data of India to estimate the roles of NPIs and viccination in the prevention and controlling of COVID-19. In the case of two patchy, we give the conditions for India and China to be open to navigation. Furthermore, an appropriate dispersal of population between India and China is obtained.

\subsection{The case of one patchy}
\noindent

Let $n=1$ in system (\ref{2.14}), it then follows from (\ref{jbzsh}) that the basic reproduction number is
\begin{equation}\label{jibenzaishshu}
\mathscr{R}_{0}=\frac{\beta_1\Lambda_1(\mu_1+\xi_1(1-\tau_1))(D_{E_{21}}(\mu_1+\frac{1}
{D_{I_{11}}})(1-C_{a1})+1-C_{s1})}{\mu_1N_1D_{E_{11}}D_{E_{21}}(\xi_1+\mu_1)(\mu_1+\frac{1}
{D_{E_{11}}})(\mu_1+\frac{1}{D_{I_{11}}})(d_1+\mu_1+\frac{1}{D_{E_{21}}})}.
\end{equation}

We estimate $d_1=0.00013$ {Day}$^{-1}$. According to the dates in \cite{population, age}, we take $N_1=1380004000$ $\mu_1=0.00004$ {Day}$^{-1}$ and $\Lambda_1=65786$ {People}/{Day}. It follows from the references \cite{huirui, moderna, janssen} that the vaccine efficacy of Pfilzer-BioNTech COVID-19 Vaccine, Moderna COVID-19 Vaccine and Janssen COVID-19 Vaccine are 95\%, 94\% and 66\%, respectively. Hence, we get the average efficacy of vaccine is 85\%. The following is the cumulative and active cases of COVID-19 in India from April 1, 2021 to June 28, 2021\cite{actualdata}. Take the average efficiency of vaccine $\tau_1=0.85$ and $\xi_1=0.05$ which is the date in Indian on April 1\cite{actualdata}. Considering the existing prevention and control intensity, we assume that $C_{a1}=C_{s1}=0.2$. According to the data of cumulative and active cases on April 1 in India, taking the initial value is $(S(0),V(0),E_1(0),E_2(0),I_1(0),I_2(0),R(0))=(1297600506,70104203,129660,102834,115020,475941,11475836)$. Thus, the numerical simulation results are shown in Figure \ref{F1}.

\begin{figure}[H]
\centering
\subfigure{
\includegraphics[height=2.4in,width=4.8in]{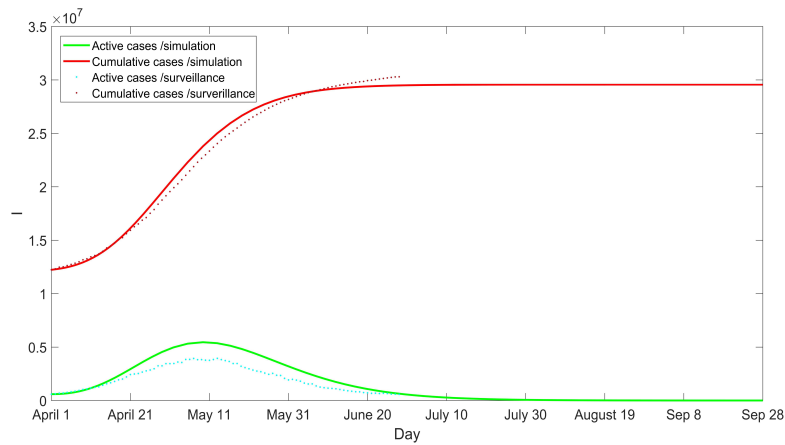}}
\caption{Numerical simulation of active and cumulative cases}
\label{F1}
\end{figure}

The discrete points represent the real data of active and cumulative cases from April 1 to June 28, 2021 in India. We found that the dynamical
results fit well with the statistical data(see Figure \ref{F1}). If the existing protection intensity and vaccine injection schedule are maintained, the numerical simulation forecasts the trend of COVID-19 In India and the active cases will reach 100000 in July 25, 10000 in  August 25 and 1000 in September 19.
\begin{figure}[H]
\centering
\subfigure{
\includegraphics[height=2.78in,width=4.5in]{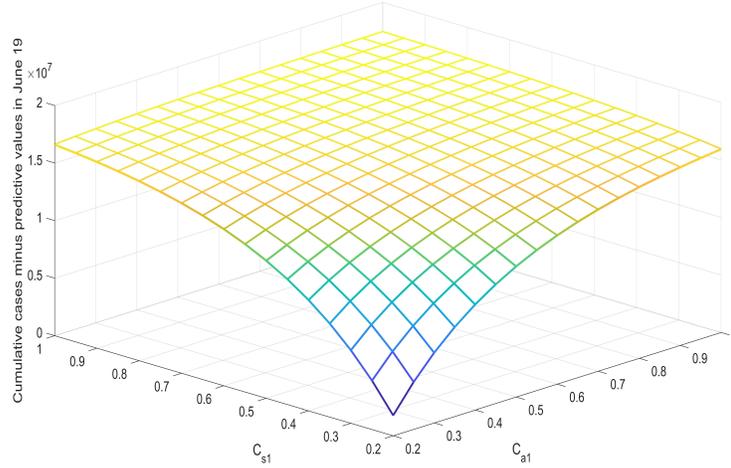}}
\caption{The decrement of numbers of infectious with the increase of $C_{s1}$ and $C_{a1}$.}
\label{F2}
\end{figure}

In the following, we investigate the role of NPIs in the controlling of COVID-19.
Keeping other parameters unchanged,
Figure \ref{F2} shows that $C_{s1}$ and $C_{a1}$ can effectively reduce the number of cumulative cases. Let $P$ be the decrement of numbers of infectious with the increase of $C_{s1}$ and $C_{a1}$. More intuitively, we have listed the specific numbers under the different intensity of NPIs(see Table 1 ).

\begin{table}[H]\label{shaojian}
\centering
\caption{The relationship between $P$ and $C_{s1}$ and $C_{a1}$}
\begin{tabular}{|l|c|c|c|c|c|c|c|c|c}
\hline
\diagbox{$C_{s1}$}{$P$}{$C_{a1}$} & 0.2 & 0.25 & 0.3 & 0.35 & 0.4 & 0.45 & 0.5 \\
\hline
0.2	&	1769611 	&	4115444 	&	6091986 	&	7757921 	&	9162980 	&	10349134 	&	 11351706 	 	 \\
\hline																	
0.25	&	4246751 	&	6219307 	&	7879821 	&	9278475 	&	10457744 	&	11453237 	&	 12294780 	 	\\
\hline																	
0.3	&	6346002 	&	8001037 	&	9393367 	&	10565718 	&	11554108 	&	12388618 	&	 13094313 	 \\
\hline																	
0.35	&	8121646 	&	9507596 	&	10673095 	&	11654365 	&	12481781 	&	13180634 	&	 13771950 	\\
\hline																	
0.4	&	9621232 	&	10779821 	&	11754026 	&	12574353 	&	13266301 	&	13851057 	&	 14346209 	 \\
\hline																	
0.45	&	10885953 	&	11853052 	&	12666328 	&	13351389 	&	13929545 	&	14418491 	&	 14832893 	\\
\hline																	
0.5	&	11951462 	&	12757688 	&	13435882 	&	14007466 	&	14490192 	&	14898792 	&	 15245468 	 \\
\hline																	
																	
\end{tabular}
\end{table}

It then follows from the dates in Table 1 if $C_{s1}$ remains unchanged and $C_{a1}$ is raised from 0.2 to 0.3, then there will be 6091986 fewer infected individual. Supposing that $C_{a1}$ remains unchanged and $C_{s1}$ is raised from 0.2 to 0.3, then there will be 6346002 fewer infected individual. If both $C_{a1}$ and $C_{s1}$ are increased to 0.3, then there will be 9393367 fewer infected individual. According to  the disease-induced mortality rate in Indian,  if we take $C_{s1}=0.2$, $C_{a1}=0.3$, then 62566 people are saved; If $C_{s1}=0.3$, $C_{a1}=0.2$, then 65174 people are saved; If $C_{s1}=0.3$ and $C_{a1}=0.3$, then 96471 people are saved.

 At present, the efficiency of vaccine is not relatively low. In the following, we assume that all people are vaccined Janssen COVID-19 Vaccine and consider the relationship between $\mathscr{R}_0$ and $C_{s1}$, $C_{a1}$(see Figure \ref{F9123}).

\begin{figure}[H]
\centering
\subfigure{
\includegraphics[height=2.78in,width=4.5in]{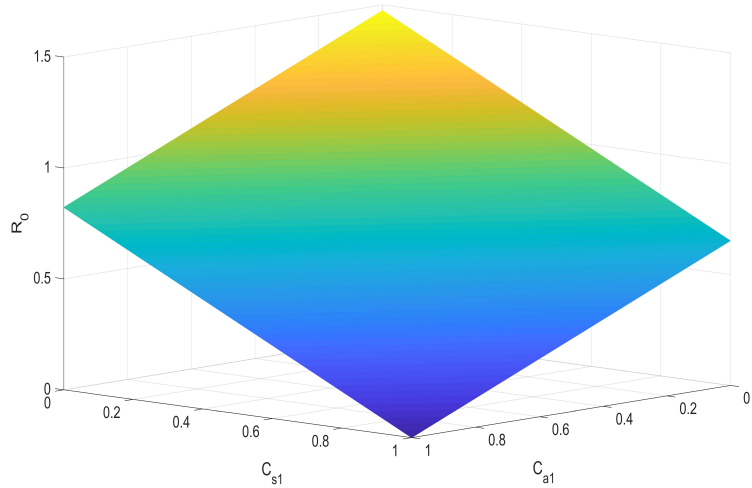}}
\caption{When $\tau_1=0.66$ and $\xi_1=1$, the image in three dimensions of relationship among $\mathscr{R}_0$, $C_{s1}$ and $C_{a1}$.}
\label{F9123}
\end{figure}

It can be seen from the above discussion that NPIs play a very significant role for the disease control.
Figure \ref{F3} is the projection of Figure \ref{F9123} on the $C_{s1}\times C_{a1}$ plane. The red and green area boundary line in Figure \ref{F3} represents $\mathscr{R}_{0}=1$. Figure \ref{F3} shows if $(C_{a1},C_{s1})$ belongs to the green area, $\mathscr{R}_{0}<1$ while $\mathscr{R}_{0}>1$ in the red area.
Our numerical results shows that NPIs are indispensable even if all the people were vaccinated when the efficiency of vaccine is relatively low. In other words, in order to control the spread of the disease, NPIs must be strengthened to make $C_{s1}$ and $C_{a1}$ in the green area $A$ even if each people is vaccinated. In particular, we suggest that NPIs should be strengthened, not weakened in India.

\begin{figure}[H]
\centering
\subfigure{
\includegraphics[height=2.4in,width=4.8in]{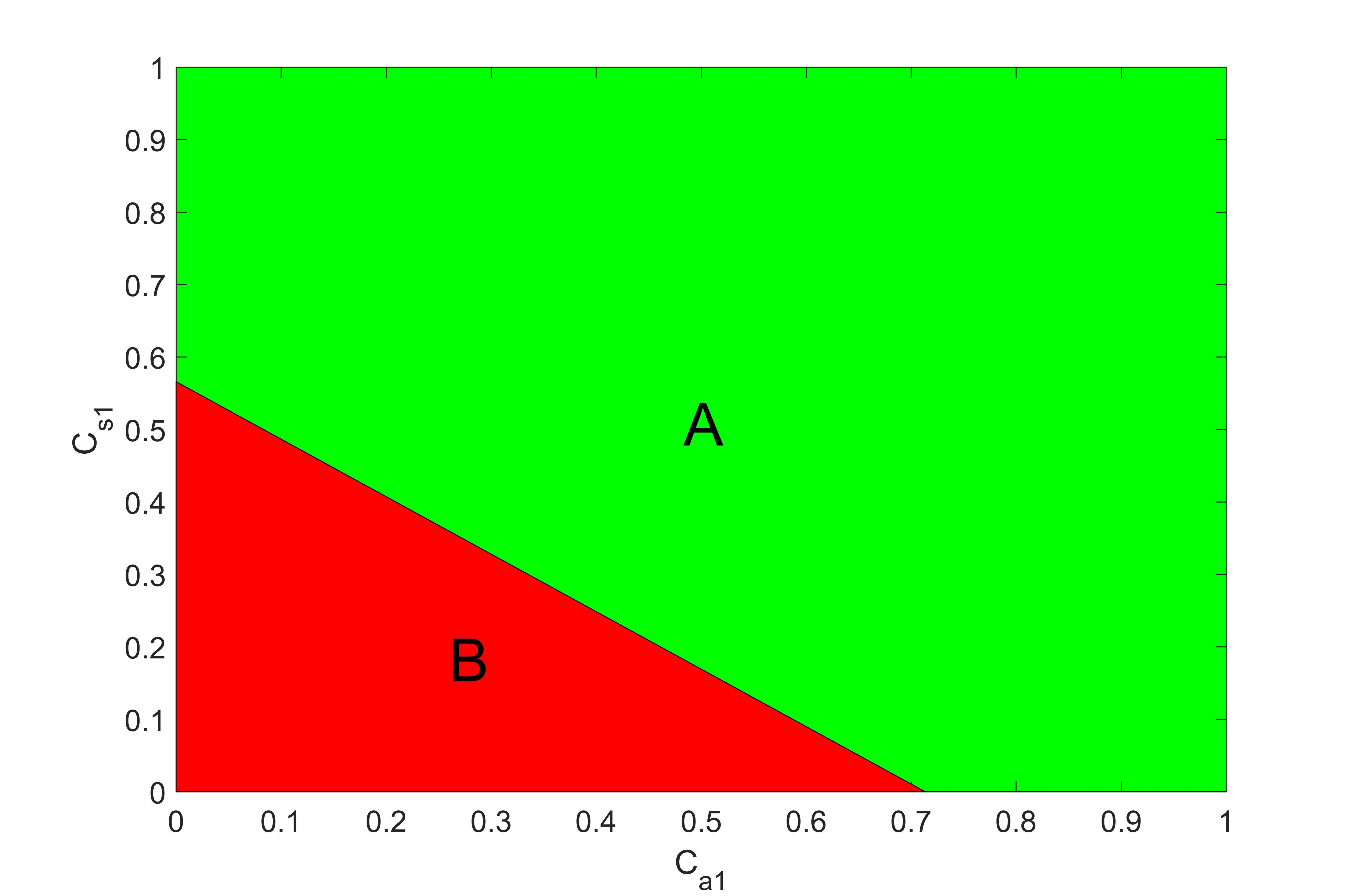}}
\caption{When $\tau_1=0.66$ and $\xi_1=1$, the relationship between $C_{s1}$ and $C_{a1}$.}
\label{F3}
\end{figure}

The herd immunity is our ultimate goal. In the following, we study the role of the vaccine in the absence of NPIs.

\begin{figure}[H]
\centering
\subfigure{
\includegraphics[height=2.4in,width=4.5in]{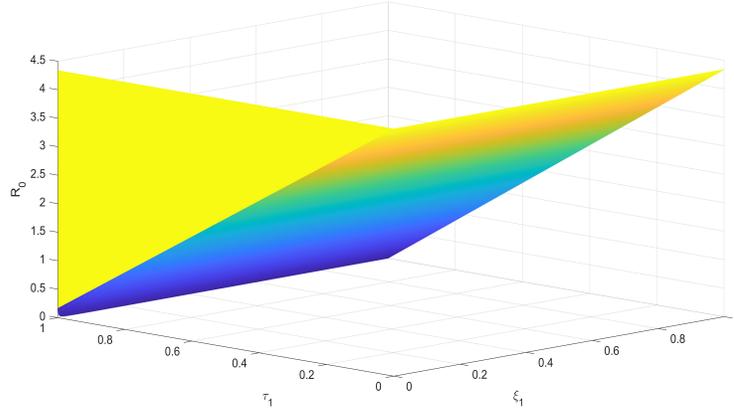}}
\caption{When $C_{s1}=0$ and $C_{a1}=0$, the relationship between $\xi_1$ and $\tau_1$ and $R_0$.}
\label{F4}
\end{figure}
It then follows from Figure \ref{F4} that $\mathscr{R}_0$ decreases with the improvement of $\xi_1$, and finally $\mathscr{R}_0$ is less than 1 when $\tau_1=0.95$. $\mathscr{R}_0$ is decreases and bigger than 1 even if $\xi_1=1$ when $\tau_1=0.66$.  Let $\mathscr{R}_0=1$, $C_{a1}=C_{s1}=0$ and $\xi_1=1$, it then follows from  (\ref{jibenzaishshu}) that $\tau_1=0.769$. In other words, $\mathscr{R}_{0}>1$ always holds when $\tau_1<0.769$. Hence, we have gotten a minimum standard of the efficiency of vaccine.

In the face of vaccine dose shortages and logistical challenges, it's impossible for all people to be vaccinated in a shorten time. It is easy see that $\tau_1\xi_1$ indicates the proportion of antibody produced after vaccination. Let $C_{a1}=C_{s1}=\xi_1=\tau_1=0$, which implies that COVID-19 transmits without NPIs and vaccines. It then follows from (\ref{jibenzaishshu}) that $\mathscr{R}_{0}=4.3356$. Furthermore, we conclude if $1-\frac{1}{\mathscr{R}_0}<\tau_1\xi_1$ is satisfied(\cite{Herd}), then the herd immunity is formed.
\begin{figure}[H]
\centering
\subfigure{
\includegraphics[height=2.78in,width=4.5in]{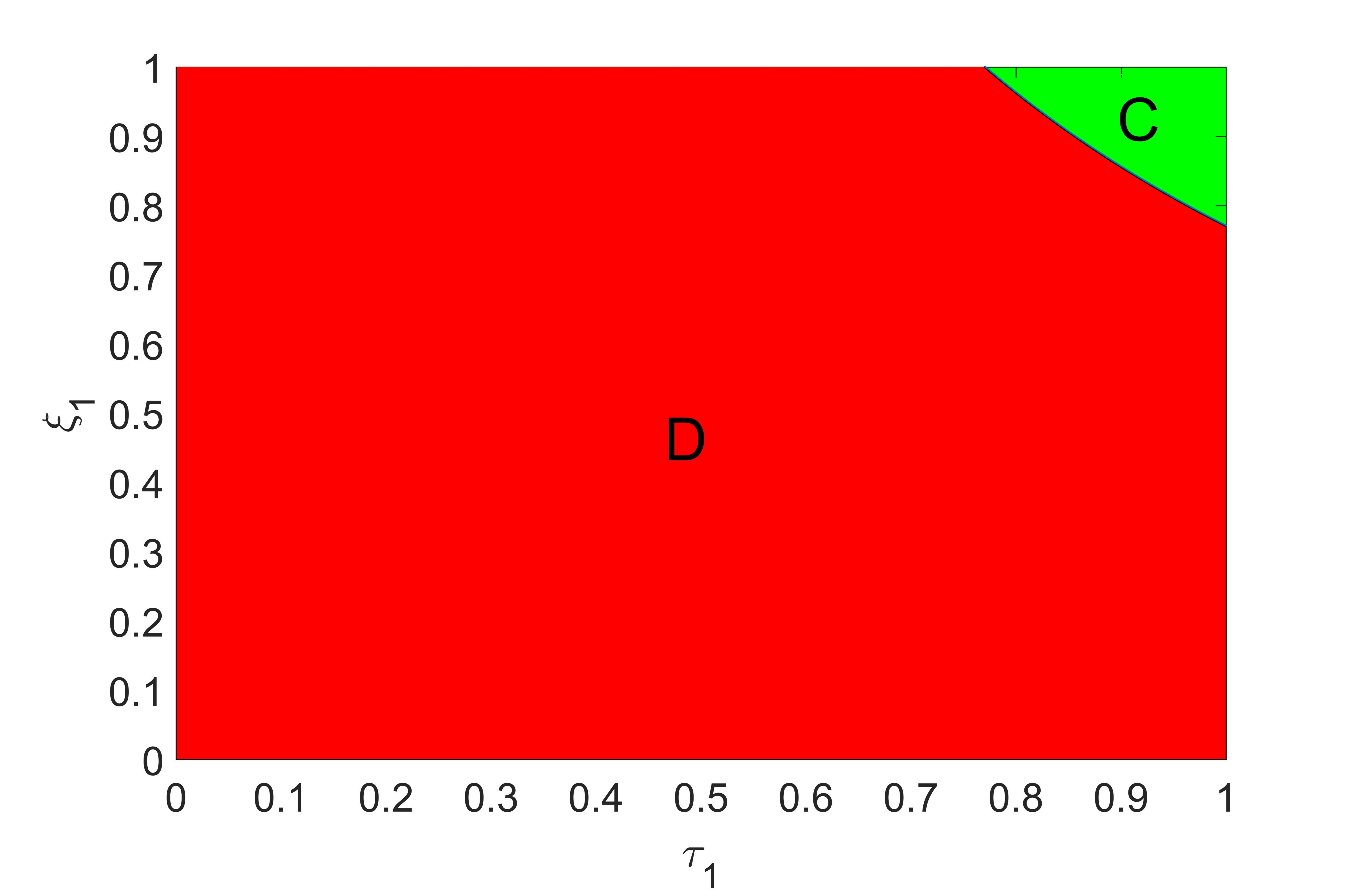}}
\caption{When $C_{a1}=0$ and $C_{s1}=0$, the relationship between herd immunity and $\tau_1$ and $\xi_1$}
\label{F6}
\end{figure}
From Figure \ref{F6}, it reveals that $\xi_1=0.769$ when $\tau_1=1$, and $\tau_1=0.769$ if $\xi_1=1$. The intersection of the area $C$ and $D$ is called the herd immunity line which satisfies $\tau_1\xi_1=0.769$ and $C$ is the herd immunity area where the condition $\tau_1\xi_1>0.769$ is satisfied.

\subsection{The case of two patches}
\noindent

In this subsection, we take India and China for example and investigate the conditions for two countries to be open to navigation. Let $n=1$ in system (\ref{2.14}) to denote the case of India, while the case of China corresponds to $n = 2$.
For simplicity, let
\begin{equation}\nonumber
\begin{split}
G=\left(\begin{array}{cc}
   G_{11} & G_{12}\\
     G_{21} & G_{22}\\
      \end{array}
      \right):=&\left(\begin{array}{cc}
   \frac{1}{D_{E_{21}}}+\mu_1+d_1-C_{11} & -C_{12}\\
    -C_{21} & \frac{1}{D_{E_{22}}}+\mu_2+d_2-C_{22}\\
      \end{array}
      \right)^{-1}\left(\begin{array}{cc}
   \frac{1}{D_{E_{21}}} & 0\\
    0 & \frac{1}{D_{E_{12}}}\\
      \end{array}\right)\times\\
      &\left(\begin{array}{cc}
   \frac{1}{D_{E_{11}}}+\mu_1-B_{11} & -B_{12}\\
     -B_{21} & \frac{1}{D_{E_{12}}}+\mu_2-B_{22}\\
      \end{array}
      \right)^{-1}\\
      H=\left(\begin{array}{cc}
   H_{1} & 0\\
     0 & H_{2}\\
      \end{array}
      \right):=&\left(\begin{array}{cc}
   \frac{1}{D_{I_{11}}}+\mu_1 & 0\\
    0 & \frac{1}{D_{I_{12}}}+\mu_2\\
      \end{array}
      \right)^{-1}
      \end{split}
      \end{equation}
	
It then follows from (\ref{jbzsh}) that

\begin{equation}\label{tworo}
  \begin{split}
    \mathscr{R}_{0}=&\frac{1}{2}((G_{11}\beta_1\frac{S^*_1+(1-\tau_1)V^*_1}{N_1}(C_{a1}-1+\frac{H_1}
    {D_{E_{21}}}(C_{s1}-1))-G_{22}\beta_2\frac{S^*_2+(1-\tau_2)V^*_2}{N_2}\times\\
    &(C_{a2}-1+\frac{H_2}{D_{E_{22}}}(C_{s2}-1)))^2+4G_{12}G_{21}
    \beta_1\beta_2\frac{(S^*_1+(1-\tau_1)V^*_1)(S^*_2+(1-\tau_2)V^*_2)}{N_1N_2}\times\\
    &(C_{a1}-1+\frac{H_1}{D_{E_{21}}}(C_{s1}-1))(C_{a2}-1+\frac{H_2}{D_{E_{22}}}
    (C_{s2}-1)))^{\frac{1}{2}}-\frac{1}{2}G_{11}\beta_1\frac{S^*_1+(1-\tau_1)V^*_1}{N_1}\times\\
    &(C_{a1}-1+\frac{H_1}{D_{E_{21}}}(C_{s1}-1))-
    \frac{1}{2}G_{22}\beta_2\frac{S^*_2+(1-\tau_2)V^*_2}{N_2}(C_{a2}-1+\frac{H_2}{D_{E_{22}}}(C_{s2}-1)),
     \end{split}
\end{equation}
where $(S^*_1,S^*_2,V^*_1,V^*_2)$ is the disease-free solution with
\begin{equation}\nonumber
  \begin{split}
    S^*_1&=\frac{-\Lambda_1\xi_2-\Lambda_1\mu_2-A_{12}\Lambda_2+A_{22}\Lambda_1}{A_{11}\xi_2+A_{22}\xi_1+A_{11}\mu_2+A_{22}\mu_1-\xi_2\mu_1-\xi_1\mu_2-\mu_1\mu_2-\xi_1\xi_2},\\
S^*_2&=\frac{-\Lambda_2\xi_1-\Lambda_2\mu_1-A_{11}\Lambda_2+A_{21}\Lambda_1}{A_{11}\xi_2+A_{22}\xi_1+A_{11}\mu_2+A_{22}\mu_1-\xi_2\mu_1-\xi_1\mu_2-\mu_1\mu_2-\xi_1\xi_2},\\
V^*_1&=\frac{-A_{12}S_2^*\xi_2+A_{22}S^*_1\xi_1-S^*_1\xi_1\mu_2}{A_{11}\mu_2+A_{22}\mu_1-\mu_1\mu_2},\\
V^*_2&=\frac{-A_{21}S_1^*\xi_1+A_{11}S^*_2\xi_2-S^*_2\xi_2\mu_1}{A_{11}\mu_2+A_{22}\mu_1-\mu_1\mu_2}.\\
  \end{split}
\end{equation}
First, we consider the conditions for free navigations. By the dates in \cite{Bureau}, let $\Lambda_2=32954$ People/Day, $N_2=1411780000$ People, $\mu_2=0.000036$ Day$^{-1}$. We estimate the $d_2=0.00016$ Day$^{-1}$. According to the dates in \cite{efficiency}, the efficacy of Sinopharm COVID-19 Vaccine is 0.73 and Tianjin CanSino COVID-19 Vaccine is 0.66. Take the average efficacy of vaccine is 0.70. We estimate that the number of people traveled from China to India is 250000 and from India to China is 1400000 annually. Thus, $A_{21}=2.78\times 10^{-6}$ Day$^{-1}$, $B_{21}=2.78\times 10^{-6}$ Day$^{-1}$, $C_{21}=2.78\times 10^{-6}$ Day$^{-1}$, $D_{21}=2.78\times 10^{-6}$ Day$^{-1}$, $A_{12}=4.85\times 10^{-7}$ Day$^{-1}$, $B_{12}=4.85\times 10^{-7}$ Day$^{-1}$, $C_{12}=4.85\times 10^{-7}$ Day$^{-1}$ and $D_{12}=4.85\times 10^{-7}$ Day$^{-1}$. In the situation of without NPIs and vaccines, i.e., the other parameters are defined as the above and $\tau=\tau_i=\xi=\xi_i=C_{ai}=C_{si}=0$, $i=1,2$, it then follows from (\ref{tworo}) that $\mathscr {R}_0=4.0843$. By the formula of the herd immunity $1-\frac{1}{\mathscr {R}_0}<\tau\xi$, the relationships between $\tau$ and $\xi$ are shown in the following(see Figure \ref{F7}).

\begin{figure}[H]
\centering
\subfigure{
\includegraphics[height=2.78in,width=4.5in]{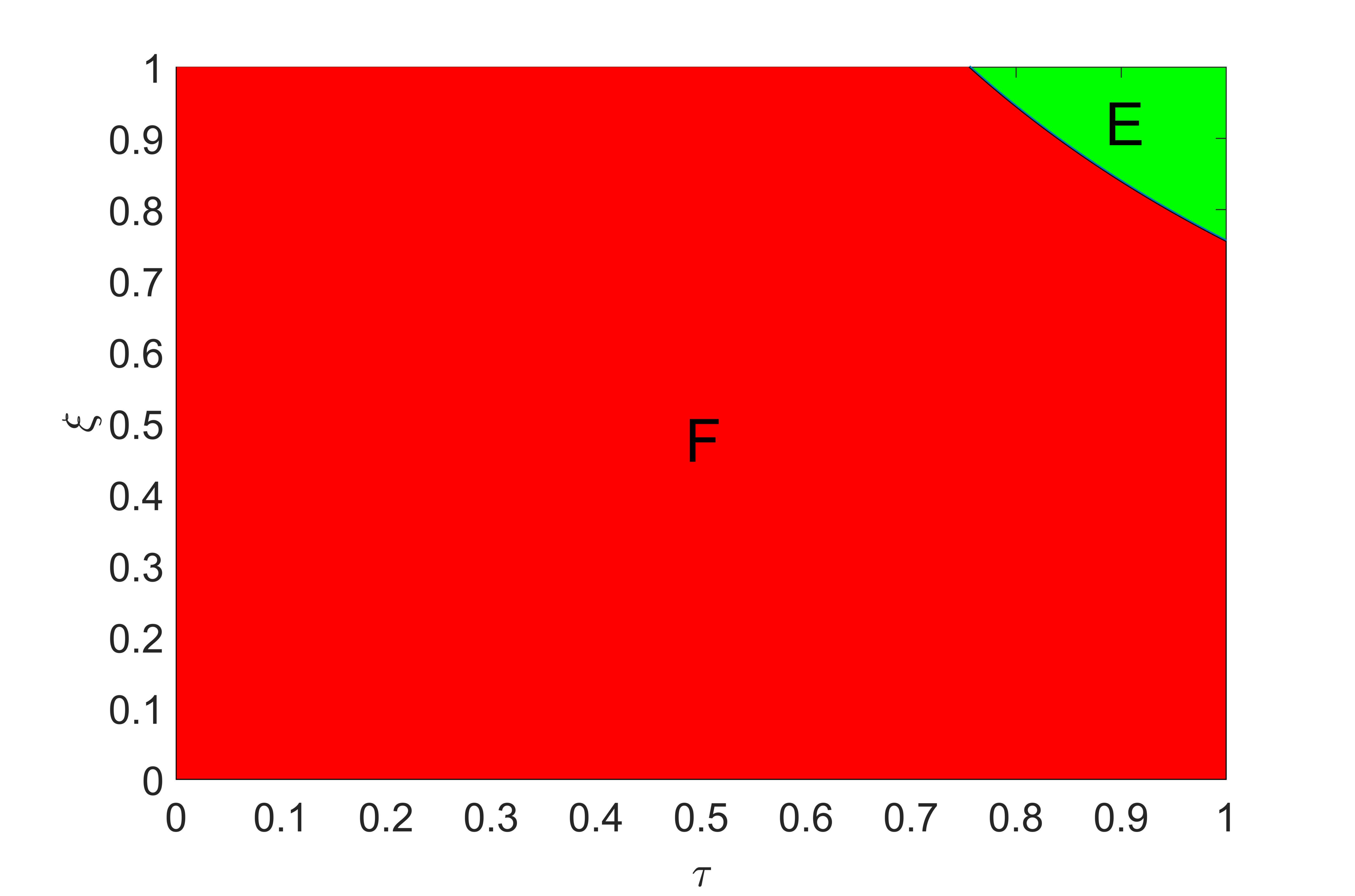}}
\caption{The relationship between herd immunity and $\xi$ and $\tau$.}
\label{F7}
\end{figure}
It is easy to see from Figure \ref{F7} that $\xi=0.755$ when $\tau=1$,  and $\tau=0.755$ if $\xi=1$. The area $E$ is the herd immunity area. The above arguments imply that it is impossible to achieve the herd immunity since the average efficacy of vaccine $\tau=0.70$. In other words, India and China do not meet the conditions for free navigations unless it is recommended to improve the efficiency of vaccinate or strengthen NPIs even if $\xi$ reaches 100\%.

In the following, we consider the influence of dispersal rate on the transmission of the disease. Let $C_{a1}=C_{s1}=0.2$, $C_{a2}=C_{s2}=0.6$, $\tau_1=0.85$, $\tau_2=0.7$, $\xi_1=0.23$, $\xi_2=0.82$ and assume $A_{21}=B_{21}=C_{21}$, $A_{12}=B_{12}=C_{12}$, the following is the relationship among $A_{21}$, $A_{12}$ and $\mathscr{R}_0$(see  Figure \ref{F8}).
\begin{figure}[H]
\centering
\subfigure{
\includegraphics[height=2.6in,width=4.5in]{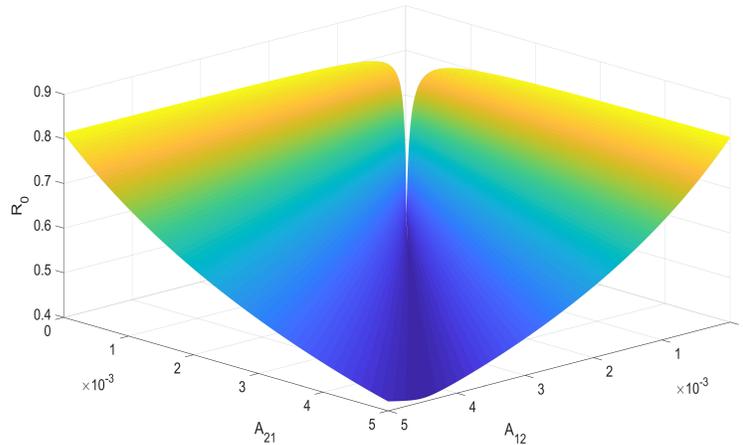}}
\caption{The relationships among $A_{21}$, $A_{12}$ and $\mathscr{R}_0$ .}
\label{F8}
\end{figure}

\begin{figure}[H]
\centering
\subfigure{
\includegraphics[height=2.4in,width=4.5in]{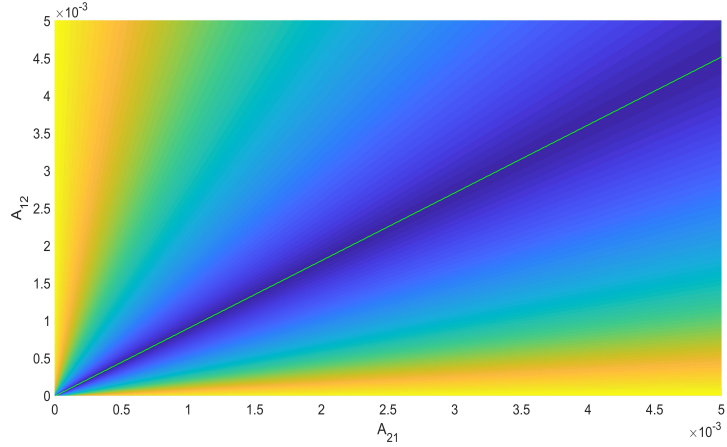}}
\caption{The projection  of Figure \ref{F8}. }
\label{F9}
\end{figure}		
Figure \ref{F9} implies that the basic reproduction number $\mathscr{R}_0$ may increase or decrease with the increase of dispersal rate. 																											 Let $A_{12}=A_{21}=0$ and the other parameters be unchanged. It then follows from (\ref{jibenzaishshu}) that $\mathscr{R}_0=0.5208$ in India and $\mathscr{R}_0=0.2831$ in China.  Let $A_{12}=4.85\times10^{-7}$ and $A_{21}=0$, then $\mathscr{R}_0=0.5247$. Let $A_{12}=0$ and $A_{21}=2.78\times10^{-6}$, then $\mathscr{R}_0=0.4870$. Let $A_{12}=4.85\times10^{-7}$ and $A_{21}=2.78\times10^{-6}$, then $\mathscr{R}_0=0.4910$. Fixed $A_{21}$, Figure \ref{F8} implies that there exists an $A_{12}$ such that $\mathscr{R}_0$ takes the minimum value. The green line in Figure \ref{F9} is the set of all such points.
Furthermore, we can  ascertain that the smallest value $\mathscr{R}_0=0.0.4035$ which is corresponded by the point in the plane of $A_{12}=2.8820*10^{-9}$ and $A_{21}=1.1636*10^{-5}$ (see Figure \ref{F9}). In other word, the appropriate dispersal of population of China to India is 1486 people, and the population of India to China is 5861071 people every year.

\section{Discussion}
\noindent

Emphasizing non-pharmaceutical interventions(NPIs) and vaccines, the dynamics of a SVEIR COVID-19 model is considered by means of the basic reproduction number. In the case of one patchy, we considered the situations in India. Our numerical result predicts that the Indian epidemic will be controlled until October if the existing intensity of NPIs and vaccine injection schedule are maintained. If the outbreak occurs repeatedly, we suggest that NPIs should be strengthened. Furthermore, it is shown that NPIs are indispensable even if all the people were vaccinated when the efficiency of vaccine is relatively low. In order to obtain the herd immunity, we speculate in numerical simulation that the minimum standard of vaccine efficiency is 76.9\%. In the face of vaccine dose shortages and logistical challenges, the herd immunity area is given. In the case of two patchyes, we conclude that India and China do not meet the conditions for free navigation under nowadays situations. In order to prevent the disease outbreaks, the appropriate dispersal implies that people in the countries or regions with serious epidemic situation should be allowed to enter into the countries or regions where the epidemic situation is mild. Certainly, the optimal strategy is that the counties affected slightly by COVID-19
supply medical supplies for the countries where COVID-19 is worst. We expect that COVID-19 will die out as soon as possible by the efforts of people of all over the world.

\end{document}